\documentclass[10pt]{amsart}
\usepackage[paper=a4paper,left=3cm,right=3cm,top=3cm,bottom=3cm]{geometry}
\usepackage{amssymb,amsthm,graphicx}

\newtheorem{theorem}{Theorem}[section]
\newtheorem{proposition}[theorem]{Proposition}

\newtheorem{claim}{Claim}

\theoremstyle{definition}
\newtheorem{definition}[theorem]{Definition}
\newtheorem{example}[theorem]{Example}

\theoremstyle{remark}

\numberwithin{equation}{section}

\newcommand{ \norm }[1]{ {\lVert #1 \rVert}_{1} }
\newcommand{ \mysum }{ \sum\nolimits }
\newcommand{ \supp }{ \operatorname{supp} }
\newcommand{ \aff }{ \operatorname{aff} }
\newcommand{ \inter }{ \operatorname{int} }
\newcommand{ \arginf }{ \operatorname{arg\,inf} }

\newcommand{ \myproof }[1]{ \noindent\textit{Proof:} #1  \hfill $\blacktriangle$%
            }
\renewcommand{\qed}{\hfill $ \blacksquare $}

\graphicspath{ {Figs/} }
\DeclareGraphicsExtensions{.eps}

\begin{document}

\title[Information-Geometric Equivalence of Transportation Polytopes]
      {Information-Geometric Equivalence of\\Transportation Polytopes}

\author{Mladen Kova\v{c}evi\'{c}}

\address{Department of Electrical Engineering, 
         Faculty of Technical Sciences, 
         University of Novi Sad, 
         Trg Dositeja Obradovi\'{c}a 6, 21000 Novi Sad, Serbia}
\email{kmladen@uns.ac.rs ({\bf corresponding author}), cet\_ivan@uns.ac.rs, vojin\_senk@uns.ac.rs}
%
\author{Ivan Stanojevi\'c}
\author{Vojin \v{S}enk}

\subjclass[2010]{94A17, 62B10, 62H17, 52B11, 52B12, 54C99}

\date{March 9, 2015.}

\keywords{Information projection, Kullback-Leibler divergence, transportation
polytope, distributions with given marginals, contingency table,
Fr\'{e}chet-Hoeffding bounds.}

\begin{abstract}
 This paper deals with transportation polytopes in the probability simplex
(that is, sets of categorical bivariate probability distributions with
prescribed marginals).
Information projections between such polytopes are studied, and a sufficient
condition is described under which these mappings are homeomorphisms.
\end{abstract}

\maketitle

\section{Preliminaries}

Let $ \Gamma_{n} $ denote the set of probability distributions with alphabet
$ \{ 1, \ldots, n \} $:
\begin{equation}
 \Gamma_{n} = \left\{ (p_i) \in \mathbb{R}^{n} \,:\, p_i \geq 0,\, \mysum_{i} p_i = 1 \right\} .
\end{equation}
The support of a probability distribution $ P = (p_i) $ is denoted by
$ \supp(P) = \{ i : p_i > 0 \} $, and its size by $ |\supp(P)| $.
The support of a set $ \mathcal P $ of probability distributions is defined as
$ \supp({\mathcal P}) = \bigcup\nolimits_{P \in {\mathcal P}} \supp(P) $.
If $ \mathcal P $ is convex, then there must exist $ P \in {\mathcal P} $
with $ \supp(P) = \supp({\mathcal P}) $.
We will also write $ P(i) $ for the masses of $ P $.

Let $ \mathcal{C}(P,Q) $ denote the set of all bivariate probability distributions
with marginals $ P \in \Gamma_{n} $ and $ Q \in \Gamma_{m} $:
\begin{equation}
 \mathcal{C}(P,Q) = \left\{ (s_{i,j}) \in \mathbb{R}^{n\times m} \,:\, 
                     s_{i,j} \geq 0, \, \mysum_{j} s_{i,j} = p_i \,,\, \mysum_{i} s_{i,j} = q_j \right\} .
\end{equation}
Such sets are special cases of the so-called transportation polytopes, and
have been studied extensively in probability, statistics, geometry, combinatorics, etc.\
(see, e.g., \cite{brualdi, fixed}).
In information-theoretic approaches to statistics, and in particular to the
analysis of (multidimensional) contingency tables, a basic role is played by
the so-called \emph{information projections}, see \cite{csiszar_shields} and
the references therein.
This motivates our study, presented in this note, of some formal properties
of information projections (I-projections for short) over domains of the form
$ {\mathcal C}(P,Q) $.
I-projections onto $ {\mathcal C}(P,Q) $ also arise in binary hypothesis
testing, see \cite{polyanskiy}.
Further information-theoretic results (in a fairly different direction)
regarding transportation polytopes can be found in \cite{kovacevic}.

Relative entropy (information divergence, Kullback-Leibler divergence) of the
distribution $ P $ with respect to the distribution $ Q $ is defined by: 
\begin{equation}
\label{eq_KL}
  D(P||Q) = \mysum_{i} p_i \log \frac{ p_i }{ q_i } ,
\end{equation}
with the conventions $ 0 \log \frac{0}{q} = 0 $ and $ p \log \frac{p}{0} = \infty $
for every $ q \geq 0 $, $ p > 0 $, being understood.
The functional $ D $ is nonnegative, equals zero if and only if $ P = Q $,
and is jointly convex in its arguments \cite{csiszar_korner}.

For a probability distribution $ S $ and a set of distributions $ \mathcal T $,
the I-projection \cite{chentsov68, chentsov, csiszarI, topsoe, csiszar_matus} of $ S $ onto
$ \mathcal T $ is defined as the unique minimizer (if it exists) of the functional
$ D(T||S) $ over all $ T \in {\mathcal T} $.
We shall study here I-projections as mappings between sets of the form
$ {\mathcal C}(P,Q) $.
Namely, let $ I_\text{proj} : {\mathcal C}(P_1,Q_1) \to {\mathcal C}(P_2,Q_2) $
be defined by:
\begin{equation}
 I_\text{proj}(S) = \arginf\displaylimits_{ T \in {\mathcal C}(P_2,Q_2) } D(T||S) .
\end{equation}
(Above and in the sequel we assume that $ P_1, P_2 \in \Gamma_n $ and
$ Q_1,Q_2 \in \Gamma_m $.)
The definition is slightly imprecise in that $ I_\text{proj}(S) $ can be
undefined for some $ S \in {\mathcal C}(P_1,Q_1) $, i.e., the domain of
$ I_\text{proj} $ can in fact be a proper subset of $ {\mathcal C}(P_1,Q_1) $.
This is overlooked for notational simplicity.
Another simplification is the omission of the dependence of the functional
$ I_\text{proj} $ on $ P_i, Q_i $; this will not cause any ambiguities.

Note that $ I_\text{proj}(S) $ is undefined only when $ D(T||S) = \infty $
for all $ T \in {\mathcal C}(P_2,Q_2) $.
If $ D(T||S) < \infty $ for some $ T \in {\mathcal C}(P_2,Q_2) $, then
existence of $ I_\text{proj}(S) $ follows easily from the properties of
$ {\mathcal C}(P_2,Q_2) $ and the convexity of $ D(\cdot||\cdot) $
\cite{csiszar_korner}.
Therefore, $ I_\text{proj}(S) $ exists if and only if there exists
$ T \in {\mathcal C}(P_2,Q_2) $ with $ \supp(T) \subseteq \supp(S) $.
Furthermore, it is clear that the I-projection is defined for all
$ S \in {\mathcal C}(P_1,Q_1) $ if and only if it is defined for all
vertices of $ {\mathcal C}(P_1,Q_1) $.

\section{Geometric equivalence of transportation polytopes}

The vertices of transportation polytopes are uniquely determined by their
supports and can be characterized as follows: 
$ U $ is a vertex of $ {\mathcal C}(P_1,Q_1) $ if and only if the associated
bipartite graph $ G_U $ with ``left" nodes $ \{ 1, \ldots, n \} $, ``right"
nodes $ \{ 1, \ldots, m \} $, and edges $ \{ (i,j) : U(i,j) > 0 \} $,
is a forest, i.e., contains no cycles \cite{klee}.
In fact, every face of the polytope $ {\mathcal C}(P_1,Q_1) $ is determined
by its support \cite{brualdi}.
Apart from identifying faces, the condition for two vertices being
adjacent can also be expressed in terms of supports, as can many other
geometric and combinatorial properties of transportation polytopes (see
\cite{deloera} and the references therein).
This motivates the following definition.

\begin{definition}
 We say that the polytopes $ {\mathcal C}(P_1,Q_1) $ and $ {\mathcal C}(P_2,Q_2) $
are \emph{geometrically equivalent} if for every $ S \in {\mathcal C}(P_1,Q_1) $
there exists $ T \in {\mathcal C}(P_2,Q_2) $ with $ \supp(S) = \supp(T) $, and
vice versa.
This is equivalent to saying that for every \emph{vertex} $ U \in {\mathcal C}(P_1,Q_1) $
there exists a \emph{vertex} $ V \in {\mathcal C}(P_2,Q_2) $ with $ \supp(U) = \supp(V) $,
and vice versa.
\hfill $ \blacktriangle $%
\end{definition}

Further justification of the term ``geometrically equivalent", in a certain
information-geometric sense, is given in Theorem \ref{homeomorphism} below.

\begin{example}
 To give an example of two geometrically equivalent transportation polytopes,
consider some $ {\mathcal C}(P_1,Q_1) $ that is generic (nondegenerate) \cite{deloera},
implying that the bipartite graphs defining its vertices are spanning trees, and
assume that $ Q_1 $ has only two masses ($ m = 2 $).
In this case for every vertex $ U \in {\mathcal C}(P_1,Q_1) $, $ G_U $ has $ n + 1 $
edges and therefore necessarily contains edges $ (i, 1) $ and $ (i, 2) $ for
some $ i \in \{ 1, \ldots, n \} $ (Fig.\ \ref{fig_graph}).
\begin{figure}[h]
 \centering
    \includegraphics[width=0.5\columnwidth]{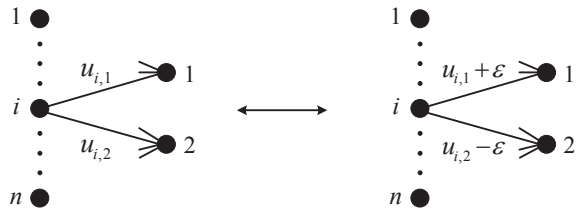}
\caption{Graphs of the vertices of $ {\mathcal C}(P_1,Q_1) $ and $ {\mathcal C}(P_1,Q_2) $.}
\label{fig_graph}
\end{figure}%
Then it is not hard to see that $ {\mathcal C}(P_1,Q_2) $ where
$ Q_2(1) = Q_1(1) + \varepsilon $, $ Q_2(2) = Q_1(2) - \varepsilon $, has
vertices with identical supports as those of $ {\mathcal C}(P_1,Q_1) $, for
small enough $ \varepsilon $.
Thus, $ {\mathcal C}(P_1,Q_1) $ and $ {\mathcal C}(P_1,Q_2) $ are
geometrically equivalent.
\hfill $ \blacktriangle $%
\end{example}

The following claim is straightforward.

\begin{proposition}
 If $ {\mathcal C}(P_1,Q_1) $ and $ {\mathcal C}(P_2,Q_2) $ are geometrically
equivalent, then they are combinatorially equivalent, i.e., they have isomorphic
face lattices.
\hfill \qed
\end{proposition}

\section{I-projections between transportation polytopes}

It is easy to see from the above discussion that $ I_\text{proj} $ maps the
vertices of $ {\mathcal C}(P_1, Q_1) $ to the vertices of $ {\mathcal C}(P_2, Q_2) $.
In the study of probability distributions with fixed marginals there are two
particularly important vertices called Fr\'{e}chet-Hoeffding (F-H for short)
upper and lower bounds \cite{fixed}.
Both of them are uniquely determined by their supports, namely,
$ \overline{F} \in {\mathcal C}(P, Q) $ is the F-H upper bound for the family
$ {\mathcal C}(P, Q) $ if and only if its associated bipartite graph has no
crossings (when the nodes $ \{1, \ldots, n\} $ on the left and $ \{1, \ldots, m\} $
on the right are drawn in increasing order, see Fig.\ \ref{fig_graph}), while
$ \underline{F} \in {\mathcal C}(P, Q) $ is the F-H lower bound if and only if
its associated graph, after reversing the order of the ``right'' nodes, has no
crossings (in other words, the support of the lower bound for $ {\mathcal C}(P, Q) $
is the same as that of the upper bound for $ {\mathcal C}(P, \tilde{Q}) $, where
$ \tilde{Q} $ is the inverse permutation of $ Q $, i.e., $ \tilde{Q}(i) = Q(m+1-i)$).
We then have:

\begin{proposition}
 Let $ \overline{F}_i, \underline{F}_i \in {\mathcal C}(P_i, Q_i) $, $ i \in \{1, 2\} $,
be the F-H upper and lower bounds.
If the I-projection of $ \overline{F}_1 $ (resp.\ $ \underline{F}_1 $) onto $ {\mathcal C}(P_2, Q_2) $
exists, it is necessarily $ \overline{F}_2 $ (resp.\ $ \underline{F}_2 $).
\hfill \qed
\end{proposition}

\noindent
Another particular case that can be derived directly is that
$ I_\text{proj}(P_1 \times Q_1) = P_2 \times Q_2 $, whenever
$ \supp(P_2) \subseteq \supp(P_1) $ and $ \supp(Q_2) \subseteq \supp(Q_1) $.
To prove this, it is by \cite[Thm 3.2]{csiszar_shields} enough to show that:
\begin{equation}
\label{eq_product}
 D(T||P_1 \times Q_1) = D(P_2 \times Q_2||P_1 \times Q_1) + D(T||P_2 \times Q_2)
\end{equation}
for all $ T \in {\mathcal C}(P_2,Q_2) $, which follows from:
\begin{equation}
  \sum_{i,j} T(i,j) \log \frac{ P_2(i) Q_2(j) }{ P_1(i) Q_1(j) }
    = D(P_2||P_1) + D(Q_2||Q_1)
    = D(P_2 \times Q_2||P_1 \times Q_1) .
\end{equation}

We now restrict our attention to geometrically equivalent polytopes.

\begin{proposition}
 $ {\mathcal C}(P_1,Q_1) $ and $ {\mathcal C}(P_2,Q_2) $ are geometrically
equivalent if and only if every vertex $ U \in {\mathcal C}(P_1,Q_1) $ has an
I-projection onto $ {\mathcal C}(P_2,Q_2) $ and every vertex $ V \in {\mathcal C}(P_2,Q_2) $
has an I-projection onto $ {\mathcal C}(P_1,Q_1) $.
\end{proposition}
\begin{proof}
 The ``only if'' part is straightforward. For the ``if'' part, take some vertex
$ U \in {\mathcal C}(P_1,Q_1) $; let its I-projection onto $ {\mathcal C}(P_2,Q_2) $
be $ U^* $, and let the I-projection of $ U^* $ onto $ {\mathcal C}(P_1,Q_1) $
be $ U' $. We know that $ \supp(U') \subseteq \supp(U^*) \subseteq \supp(U) $,
but in fact none of the inclusions can be strict because there can be no two
vertices of a transportation polytope such that the support of one of them
contains the support of the other.
\end{proof}

The main result that we wish to report in this note is stated in the
following theorem.
It is a direct consequence of the propositions proved subsequently.

\begin{theorem}
\label{homeomorphism}
 If $ {\mathcal C}(P_1,Q_1) $ and $ {\mathcal C}(P_2,Q_2) $ are geometrically
equivalent, then they are homeomorphic under information projections.
\hfill \qed
\end{theorem}

We first give a simple proof of continuity of information projections by
using a well known identity obeyed by these functionals.
See also \cite{cont} for a slightly different proof (obtained for the more
general notion of $ f $-projections). The assumed topology is the one induced
by the $ \ell_1 $ norm, and in what follows $ P_n \to P $ means that
$ \norm{ P_n - P } \equiv \mysum_i | P_n(i) - P(i) | \to 0 $.

\begin{proposition}
\label{continuous}
 $ I_\textnormal{proj} $ is continuous in its domain.
\end{proposition}
\begin{proof}
 Let $ S_n, S \in {\mathcal C}(P_1,Q_1) $ with $ S_n \to S $.
Let $ S^* = I_\text{proj}(S) $, $ S_n^* = I_\text{proj}(S_n) $; we need to
show that $ S_n^* \to S^* $.
Since $ {\mathcal C}(P_2,Q_2) $ is compact, $ S_n^* $ must have a convergent
subsequence $ S_{k_n}^* $ ($ k_n $ is an increasing function in $ n $).
Suppose that $ S_{k_n}^* \to R $ for some $ R \in {\mathcal C}(P_2,Q_2) $.
The set of all distributions $ T \in {\mathcal C}(P_2,Q_2) $ with
$ \supp(T) \subseteq \supp(S_{k_n}) $ is a linear family\footnote{\,A linear
family of (two-dimensional) probability distributions is a set of the form
$ \big\{ T : \sum_{i,j} T(i,j) f_k(i,j) = \alpha_k \big\} $,
where $ f_k $, $ 1 \leq k \leq K $, are real functions defined on the alphabet
of the distributions $ T $, and $ \alpha_k $ are real numbers.} \cite{csiszar_shields},
and therefore the following identity holds \cite[Thm 3.2]{csiszar_shields}:
\begin{equation}
 D(T||S_{k_n}) = D(S_{k_n}^*||S_{k_n}) + D(T||S_{k_n}^*)
\end{equation}
for all $ T \in {\mathcal C}(P_2,Q_2) $ with $ \supp(T) \subseteq \supp(S_{k_n}) $.
Taking the limit when $ n \to \infty $ and using the fact that $ D(\cdot||\cdot) $
is continuous in its second argument (in the finite alphabet case), we obtain:
\begin{equation}
\label{eq_limDpom}
 D(T||S) = \lim_{n \to \infty} D(S_{k_n}^*||S_{k_n}) + D(T||R) .
\end{equation}
Evaluating \eqref{eq_limDpom} at $ T = R $ we conclude that
$ \lim_{n \to \infty} D(S_{k_n}^*||S_{k_n}) = D(R||S) $.
Substituting this back into \eqref{eq_limDpom} and evaluating at $ T = S^* $ we get:
\begin{equation}
 D(S^*||S) = D(R||S) + D(S^*||R) ,
\end{equation}
wherefrom $ D(S^*||S) \geq D(R||S) $. But since $ S^* $ is by assumption
the unique minimizer of $ D(\cdot||S) $ over $ {\mathcal C}(P_2,Q_2) $, we
must have $ R = S^* $.
\end{proof}

\begin{proposition}
\label{bijection}
 Let $ {\mathcal C}(P_1,Q_1) $ and $ {\mathcal C}(P_2,Q_2) $ be geometrically
equivalent.
Then $ I_\textnormal{proj} $ is a bijection\footnote{\,Note that this follows from a
stronger statement given in Proposition \ref{inverse}, but we also give here
a direct proof that we believe is interesting in its own right.}.
\end{proposition}
\begin{proof}
 1.) $ I_\text{proj} $ is injective (one-to-one).
Observe that every distribution $ S \in {\mathcal C}(P_1,Q_1) $ maps to a distribution
with the same support, i.e., $ \supp(I_\text{proj}(S)) = \supp(S) $; this follows
from \cite[Thm 3.1]{csiszar_shields} (that such a distribution exists follows
from geometric equivalence of $ {\mathcal C}(P_1,Q_1) $ and $ {\mathcal C}(P_2,Q_2) $).
We conclude that a vertex $ V \in {\mathcal C}(P_1,Q_1) $ maps to the corresponding
vertex $ V^* \in {\mathcal C}(P_2,Q_2) $ with $ \supp(V^*) = \supp(V) $, and no other
distribution from $ {\mathcal C}(P_1,Q_1) $ can map to $ V^* $ because vertices are
uniquely determined by their supports.
Assume now, for the sake of contradiction, that
$ I_\text{proj}(S_1) = I_\text{proj}(S_2) = S^* $, where
$ S_1, S_2 \in {\mathcal C}(P_1,Q_1) $ are not vertices.
As commented above, we necessarily have $ \supp(S_1) = \supp(S_2) = \supp(S^*) $.
Furthermore, by \cite[Thm 3.2]{csiszar_shields} we have:
\begin{equation}
 \begin{aligned}
  D(T||S_1) &= D(S^*||S_1) + D(T||S^*) \\
  D(T||S_2) &= D(S^*||S_2) + D(T||S^*)
 \end{aligned}
\end{equation}
and by subtracting these equations we get:
\begin{equation}
\label{eq_injecdiff}
 D(T||S_1) - D(T||S_2) = D(S^*||S_1) - D(S^*||S_2)
\end{equation}
for all $ T \in {\mathcal C}(P_2,Q_2) $ with $ \supp(T) \subseteq \supp(S^*) $.
By writing out all terms of \eqref{eq_injecdiff} we obtain:
\begin{equation}
\label{eq_TS}
 \sum_{i,j} T(i,j) \log \frac{ S_2(i,j) }{ S_1(i,j) }
   = \sum_{i,j} S^*(i,j) \log \frac{ S_2(i,j) }{ S_1(i,j) } .
\end{equation}
Define $ \epsilon(i,j) = S_2(i,j) - S_1(i,j) $.
We can evaluate \eqref{eq_TS} at $ T = S^* + \delta \epsilon $ for some
small enough constant $ \delta > 0 $, because
$ \sum_i \epsilon(i,j) = \sum_j \epsilon(i,j) = 0 $ and $ \supp(S^*) = \supp(S_1) = \supp(S_2)$,
which ensures that $ S^* + \delta \epsilon \in {\mathcal C}(P_2,Q_2) $ and
$ \supp(T) \subseteq \supp(S^*) $.
This gives:
\begin{equation}
\label{eq_injecfinal}
 \sum_{i,j} \epsilon(i,j) \log \frac{ S_2(i,j) }{ S_1(i,j) } = 0 .
\end{equation}
But $ \epsilon(i,j) $ and $ \log \frac{ S_2(i,j) }{ S_1(i,j) } $ always
have the same sign, which means that the left-hand side of \eqref{eq_injecfinal}
is strictly positive and cannot equal zero, a contradiction.

 2.) $ I_\text{proj} $ is surjective (onto).
Let $ {\mathcal F}_k \subseteq {\mathcal C}(P_1,Q_1) $ be a $ k $-dimensional
face of $ {\mathcal C}(P_1,Q_1) $, $ k \leq (n - 1)(m - 1) $, determined
uniquely by its support $ \supp({\mathcal F}_k) $, namely,
$ {\mathcal F}_k = \{ S \in {\mathcal C}(P_1,Q_1) : 
                      \supp(S) \subseteq \supp({\mathcal F}_k) \} $.
We can regard $ {\mathcal F}_k $ as a convex and compact subset of its affine
hull, denoted $ \aff({\mathcal F}_k) $.
When regarded this way, the interior of $ {\mathcal F}_k $ is nonempty and
consists of distributions with full support, namely,
$ \inter({\mathcal F}_k) = \{ S \in {\mathcal F}_k : 
                              \supp(S) = \supp({\mathcal F}_k) \} $.
The boundary of $ {\mathcal F}_k $, denoted $ \partial {\mathcal F}_k $, is
the union of the proper faces of $ {\mathcal F}_k $.
Distributions in $ \partial {\mathcal F}_k $ have supports strictly contained
in $ \supp({\mathcal F}_k) $.
Now, let $ {\mathcal F}_k^* $ be the corresponding face of $ {\mathcal C}(P_2,Q_2) $
with $ \supp({\mathcal F}_k^*) = \supp({\mathcal F}_k) $.
We know that $ I_\text{proj} $ maps distributions from $ {\mathcal F}_k $ to
distributions from $ {\mathcal F}_k^* $
($ I_\text{proj}({\mathcal F}_k) \subseteq {\mathcal F}_k^* $) because, for
$ S \in {\mathcal F}_k $, $ D(\cdot||S) $ is finite only over $ {\mathcal F}_k^* $.
We will show that in fact $ I_\text{proj}({\mathcal F}_k) = {\mathcal F}_k^* $,
i.e., that $ I_\text{proj} $ is surjective over $ {\mathcal F}_k $, which will
establish the desired claim.
The proof is by induction on the dimension of the faces ($ k $).
We first observe, again by analyzing supports, that
$ I_\text{proj}(\inter({\mathcal F}_k)) \subseteq \inter({\mathcal F}_k^*) $, and
$ I_\text{proj}(\partial {\mathcal F}_k) \subseteq \partial {\mathcal F}_k^* $
(in fact, the image of every proper face of $ {\mathcal F}_k $ is contained in
the corresponding face of $ {\mathcal F}_k^* $ having the same support).
We can now start the induction.
Assume that $ I_\text{proj} $ is surjective over every face of
$ {\mathcal C}(P_1,Q_1) $ of dimension $ < k $
(we know that it is surjective over zero-dimensional faces, i.e., vertices,
and so the induction is justified).
Therefore, the assumption is that
$ I_\text{proj}(\partial {\mathcal F}_k) = \partial {\mathcal F}_k^* $,
and we need to show that also
$ I_\text{proj}(\inter({\mathcal F}_k)) = \inter({\mathcal F}_k^*) $.
We will use the following simple claim.

\begin{center}
\begin{minipage}{0.92\textwidth}
\begin{claim}
\label{claim}
 Let $ A $ and $ B $ be open sets (in arbitrary topological space) with
$ A \subseteq B $, and $ B $ connected.
If $ A $ and $ B $ have the same boundaries ($ \partial A = \partial B $)
then they are equal.
\end{claim}
 \myproof{Assume that $ A \neq B $, and let $ x \in B \setminus A $.
There must exist a neighborhood of $ x $, denoted $ V(x) $, such that
$ V(x) \subseteq B \setminus A $ for otherwise we would have that
$ x \in \partial A = \partial B $ which is impossible since $ B $ is open
and cannot contain its boundary points.
This proves that $ B \setminus A $ is open and hence $ B $ is a union of
two disjoint open sets ($ A $ and $ B \setminus A $).
This is a contradiction because $ B $ is connected.}
\end{minipage}
\end{center}
\noindent
We know that
$ I_\text{proj}(\inter({\mathcal F}_k)) \subseteq \inter({\mathcal F}_k^*) $,
and that $ \inter({\mathcal F}_k^*) $ is open (in $ \aff({\mathcal F}_k^*) $)
and connected.
Hence, to prove that
$ I_\text{proj}(\inter({\mathcal F}_k)) = \inter({\mathcal F}_k^*) $
(by using Claim \ref{claim}), we need to show that
$ I_\text{proj}(\inter({\mathcal F}_k)) $ is open, and that
$ \partial I_\text{proj}(\inter({\mathcal F}_k)) =
\partial \inter({\mathcal F}_k^*) \equiv \partial {\mathcal F}_k^* $.
Since $ I_\text{proj} $ is an injective and continuous function from a
compact to a metric space, it is a homeomorphism onto its image
\cite[Thm 7.8, Ch I]{bredon}.
In particular, it is both open and closed.
Therefore, $ I_\text{proj}(\inter({\mathcal F}_k)) $ is indeed open in
$ \aff({\mathcal F}_k^*) $.
Furthermore, $ I_\text{proj}({\mathcal F}_k) = 
I_\text{proj}(\inter({\mathcal F}_k)) \cup \partial {\mathcal F}_k^* $ is
closed in $ \aff({\mathcal F}_k^*) $, which implies that the boundary of
$ I_\text{proj}(\inter({\mathcal F}_k)) $ is contained in
$ \partial {\mathcal F}_k^* $.
But in fact it must be equal to $ \partial {\mathcal F}_k^* $ because any
$ T^* \in \partial {\mathcal F}_k^* $ is a limit point of
$ I_\text{proj}(\inter({\mathcal F}_k)) $.
Namely, $ T^* $ must be the image of some $ T \in \partial {\mathcal F}_k $
by the induction hypothesis, and if $ T_n \to T $, $ T_n \in \inter({\mathcal F}_k) $,
then $ I_\text{proj}(T_n) \to T^* $ by continuity.
The proof is complete.
\end{proof}

In the above proof we used the fact that the inverse of the I-projection
from $ {\mathcal C}(P_1,Q_1) $ to $ {\mathcal C}(P_2,Q_2) $ is continuous.
The following proposition precisely identifies this inverse.
The statement is somewhat counterintuitive due to the asymmetry of the
functional $ D(\cdot||\cdot) $.

\begin{proposition}
\label{inverse}
 Let $ {\mathcal C}(P_1,Q_1) $ and $ {\mathcal C}(P_2,Q_2) $ be geometrically
equivalent.
Then the inverse of the I-projection from $ {\mathcal C}(P_1,Q_1) $ to
$ {\mathcal C}(P_2,Q_2) $ is the I-projection from $ {\mathcal C}(P_2,Q_2) $
to $ {\mathcal C}(P_1,Q_1) $.
\end{proposition}
\begin{proof}
 The linear families $ {\mathcal C}(P_1,Q_1) $ and $ {\mathcal C}(P_2,Q_2) $
are translates\footnote{\,Linear families are translates of each other if they
are defined by the same functions $ f_k $ but different numbers $ \alpha_k $.}
of each other in the sense of \cite{csiszar_shields}.
Let $ S \in {\mathcal C}(P_1,Q_1) $ and let $ S^* $ be its I-projection onto
$ {\mathcal C}(P_2,Q_2) $.
By \cite[Lemma 4.2]{csiszar_shields}, the I-projections of $ S $ and $ S^* $
onto $ {\mathcal C}(P_1,Q_1) $ must be identical, and this is trivially $ S $.
(Apart from being translates of each other, the additional condition of
\cite[Lemma 4.2]{csiszar_shields} dealing with supports is also satisfied due to
geometric equivalence of $ {\mathcal C}(P_1,Q_1) $ and $ {\mathcal C}(P_2,Q_2) $.)
\end{proof}

We conclude the paper by illustrating that the converse of Theorem
\ref{homeomorphism} does not hold.
The following example exhibits two transportation polytopes that are not
geometrically equivalent, but are homeomorphic under information projection.

\begin{example}
 Let $ P_1 = (1/2, 1/2) $, $ Q_1 = (1/3, 2/3) $, and $ P_2 = Q_2 =  (1/2, 1/2) $.
Both $ {\mathcal C}(P_1,Q_1) $ and $ {\mathcal C}(P_2,Q_2) $ are one-dimensional
polytopes, but clearly not geometrically equivalent because their vertices are:
\begin{equation}
 U_1 = \begin{pmatrix}
        1/3  &  1/6  \\
        0    &  1/2
       \end{pmatrix} ,
\quad
 U_2 = \begin{pmatrix}
        0    &  1/2  \\
        1/3  &  1/6
       \end{pmatrix}
\end{equation}
for $ {\mathcal C}(P_1,Q_1) $ and
\begin{equation}
 V_1 = \begin{pmatrix}
        1/2  &  0  \\
        0    &  1/2
       \end{pmatrix} ,
\quad
 V_2 = \begin{pmatrix}
        0    &  1/2  \\
        1/2  &  0
       \end{pmatrix}
\end{equation}
for $ {\mathcal C}(P_2,Q_2) $.
Let $ I_\text{proj} $ denote the I-projection from $ {\mathcal C}(P_1,Q_1) $ to
$ {\mathcal C}(P_2,Q_2) $, as before.
$ I_\text{proj} $ is continuous by Proposition \ref{continuous}.
By using \cite[Lemma 4.2]{csiszar_shields} in the same way as in Proposition
\ref{inverse}, one can show that it is bijective over the interior of
$ {\mathcal C}(P_1,Q_1) $ (which consists of distributions from $ {\mathcal C}(P_1,Q_1) $
having full support), and that its inverse over this domain is precisely the
I-projection from $ {\mathcal C}(P_2,Q_2) $ to $ {\mathcal C}(P_1,Q_1) $.
Since $ I_\text{proj}(U_i) = V_i $, $ i \in \{1, 2\} $, $ I_\text{proj} $ is
bijective over the entire $ {\mathcal C}(P_1,Q_1) $, and hence it is a
homeomorphism.
Its inverse is guaranteed to be continuous by \cite[Thm 7.8, Ch I]{bredon},
but note that this inverse is \emph{not} the I-projection from $ {\mathcal C}(P_2,Q_2) $
to $ {\mathcal C}(P_1,Q_1) $ because the I-projection of $ V_i $ onto
$ {\mathcal C}(P_1,Q_1) $ is undefined.
\hfill $ \blacktriangle $%
\end{example}

\section*{Acknowledgment}

This work was supported by the Ministry of Education, Science and Technological
Development of the Republic of Serbia (grants TR32040 and III44003).

\bibliographystyle{amsplain}

\begin{thebibliography}{15}

\bibitem{bredon}
   G. E. Bredon,
   \emph{Topology and Geometry},
   Springer-Verlag, 1993.
\bibitem{brualdi}
   R. A. Brualdi, 
   \emph{Combinatorial Matrix Classes}, 
   Cambridge University Press, 2006.
\bibitem{chentsov68}
   N. N. Chentsov,
   ``A Nonsymmetric Distance Between Probability Distributions, Entropy and the Pythagorean Theorem,''
   \emph{Math. Notes}, vol. 4, no. 3, pp. 686-691, Sept. 1968.
\bibitem{chentsov}
   N. N. Chentsov,
   \emph{Statistical Decision Rules and Optimal Inference},
   (in Russian). Providence, RI: Translations of Mathematical Monographs, Amer. Math. Soc., 1982.
   Original publication: Moscow, U.S.S.R.: Nauka, 1972.
\bibitem{csiszarI}
   I. Csisz\'ar,
   ``I-Divergence Geometry of Probability Distributions and Minimization Problems,"
   \emph{Ann. Probab.}, vol. 3, no. 1, pp. 146--158, 1975.
\bibitem{csiszar_korner}
   I. Csisz\'ar and J. K\"orner, 
   \emph{Information Theory: Coding Theorems for Discrete Memoryless Systems}, 
   Academic Press, Inc., 1981. 
\bibitem{csiszar_shields}
   I. Csisz\'ar and P. Shields,
   ``Information Theory and Statistics: A Tutorial,"
   \emph{Foundations and Trends in Communications and Information Theory}, vol. 1, no. 4, pp. 417--528, Dec. 2004.
\bibitem{csiszar_matus}
   I. Csisz\'ar and F. Mat\'u\v{s},
   ``Information Projections Revisited,"
   \emph{IEEE Trans. Inform. Theory}, vol. 49, no. 6, pp. 1474--1490, June 2003.
\bibitem{deloera}
   J. A. De Loera and E. D. Kim,
   ``Combinatorics and Geometry of Transportation Polytopes: An Update,"
   In \emph{Discrete Geometry and Algebraic Combinatorics}, vol. 625 of \emph{Contemporary Mathematics}, pp. 37--76, American Mathematical Society, Providence, RI, 2014.
\bibitem{cont}
   C. Gietl and F. P. Reffel,
   ``Continuity of f-projections on Discrete Spaces,"
   in \emph{Geometric Science of Information, Lecture Notes in Computer Science}, vol. 8085, pp 519--524, 2013.
\bibitem{klee}
   V. Klee and C. Witzgall,
   ``Facets and vertices of transportation polytopes,''
   In: Mathematics of the Decision Sciences, Part I (Stanford, CA, 1967), 257--282, AMS, Providence, RI, 1968.
\bibitem{kovacevic}
   M. Kova\v{c}evi\' c, I. Stanojevi\' c, and V. \v{S}enk, 
   ``On the Entropy of Couplings," 
   \emph{Inform. and Comput.}, vol. 242, pp. 369-382, June 2015. 
\bibitem{polyanskiy}
   Y. Polyanskiy,
   ``Hypothesis Testing via a Comparator,"
   in \emph{Proc. 2012 IEEE Int. Symp. Inf. Theory (ISIT)}, pp. 2206--2210, Cambridge, MA, July 2012.
\bibitem{fixed}
   L. R\"uschendorf, B. Schweizer, and M. D. Taylor (Editors), 
   \emph{Distributions with Fixed Marginals and Related Topics}, 
   Lecture Notes - Monograph Series, Institute of Mathematical Statistics, 1996. 
\bibitem{topsoe}
   F. Tops\o e,
   ``Information Theoretical Optimization Techniques,"
   \emph{Kybernetika}, vol. 15, no. 1, pp. 8--27, 1979.
%
\end{thebibliography}

\end{document}